\documentclass[12pt]{article}

\usepackage[letterpaper, margin=1in]{geometry}
\usepackage{setspace}
\usepackage{tikz}
\usepackage{tabularx}

\usepackage[T1]{fontenc}
\usepackage[utf8]{inputenc}
\usepackage{lmodern}
\usepackage{amsmath, amssymb, amsthm, mathtools}
\usepackage{dsfont}

\usepackage{graphicx}
\usepackage{subcaption}
\usepackage{xcolor}
\usepackage{booktabs}
\usepackage{multirow}
\usepackage{enumitem}

\setlist{noitemsep}
\setlist{nosep}

\graphicspath{{Figures/}{./Figures/}}

\usepackage[numbers, square, comma, sort&compress]{natbib}

\newtheorem{lemma}{Lemma}

\newtheorem{proposition}{Proposition}
\theoremstyle{definition}

\newcommand{\Pprob}{\mathbb{P}}

\newcommand{\given}{\, \vert \,}

\mathtoolsset{showonlyrefs}

\usepackage{hyperref}
\hypersetup{colorlinks=true, linkcolor=black, citecolor=blue, urlcolor=blue}
\usepackage[capitalise, noabbrev]{cleveref}

\title{On the Equivalence of Instantaneous and Mechanistic Reproduction Numbers}
\author{Jeremy Goldwasser\thanks{Department of Statistics, University of California, Berkeley.} \and  Ryan J.\ Tibshirani\footnotemark[1] \and Alyssa Bilinski\thanks{Departments of Health Policy and Biostatistics, Brown University.}}
\date{}

\begin{document}
\maketitle

\begin{abstract}
The effective reproduction number ($R_t$) is widely used to track epidemic dynamics in real time.
The standard estimation framework uses ``instantaneous $R_t$,'' defined via the renewal equation, which relates new infections to past infections through a generation interval distribution.
Compartmental models like SEIR yield a seemingly distinct quantity---``mechanistic $R_t$''---based on the effective contact rate and duration of infectiousness.
We prove these two definitions are equivalent under homogeneous mixing, the standard assumption in compartmental modeling.
We also derive the generation interval distribution implied by SEIR dynamics.
A practical consequence is that generation intervals, often treated as assumption-light inputs to renewal equation estimators, in fact encode specific compartmental structure.
\end{abstract}

\section{Introduction}

The effective reproduction number ($R_t$) is a central metric in infectious disease epidemiology. It summarizes the transmissibility of a pathogen as the average number of secondary infections generated by a single infected individual at time $t$.
Reliable, timely estimates of $R_t$
are useful for public health decision-making, enabling policymakers to adopt and adjust intervention measures and assess their effectiveness.

There are multiple competing definitions of $R_t$, which can be operationalized in several ways \citep{gostic2020practical}.
Two were introduced in \citet{fraser2007}.
\textit{Case $R_t$} is the average number of secondary transmissions amongst the cohort of individuals infected at time $t$: a forward-looking quantity.
\textit{Instantaneous $R_t$} instead describes the average number of transmissions that occur at $t$, from infections before $t$.
This backward-looking notion is the standard definition in practice: it engenders a formula called the renewal equation, which is amenable to real-time estimation with methods like EpiEstim and EpiNow2 \citep{cori2013new, abbott2020estimating}.

A third formulation arises from compartmental epidemic models.
These models decompose a population into groups based on disease status \citep{anderson1991infectious, kermack1927contribution}.
For example, SEIR models place individuals as Susceptible, Exposed, Infectious, or Recovered.
Movement between compartments is governed by transmission and recovery parameters; the reproduction number can be expressed in terms of these.
Models like SEIR are more commonly used in the broader infectious disease modeling literature for intervention design, scenario planning, and forecasting.
However, the link between renewal equation $R_t$ estimators and standard mechanistic models is less well-established.

In this paper, we develop an explicit crosswalk between these two perspectives.
We show that if both the mechanistic model and generation interval are correctly specified, the definition of $R_t$ that
arises from compartmental dynamics equals the instantaneous $R_t$ targeted by the renewal equation.
We also derive the generation interval implicit in SEIR models.

\subsection*{Notation}

\cref{tab:notation} summarizes the notation used throughout this paper.

\begin{table}[!ht]
\centering
\caption{Notation used throughout this paper.}
\label{tab:notation}
\renewcommand{\arraystretch}{1.15}
\begin{tabularx}{\linewidth}{@{}l >{\raggedright\arraybackslash}X@{}}
\toprule
\textbf{Symbol} & \textbf{Meaning} \\
\midrule
\multicolumn{2}{@{}l}{\textsc{Reproduction numbers}} \\
\addlinespace[2pt]
$R_0$ & Basic reproduction number \\
$R_t^I$ & Instantaneous reproduction number \\
$R_t^C$ & Case (cohort) reproduction number \\
$R_t^M$ & Mechanistic reproduction number \\
\midrule
\multicolumn{2}{@{}l}{\textsc{Infection counts}} \\
\addlinespace[2pt]
$x_t$ & Infections at $t$ ($E_t^*$ in SEIR model) \\
\midrule
\multicolumn{2}{@{}l}{\textsc{Renewal structure}} \\
\addlinespace[2pt]
$w(t,k)$ & Rate of secondary transmission at $t$ from an infection at $t-k$ \\
$g$ & Generation interval distribution \\
\midrule
\multicolumn{2}{@{}l}{\textsc{Compartmental structure}} \\
\addlinespace[2pt]
$N$ & Population size \\
$S_t, E_t, I_t$ & Susceptible, Exposed, Infectious prevalence at time $t$ \\
$E_t^*, I_t^*, R_t^*$ & Incidence into $E$, $I$, $R$ compartments at $t$ \\
$\beta_t$ & Effective contact rate \\
$\sigma,\ \gamma$ & Exit rates from $E$ and $I$ in basic compartmental model  \\
$ \mu^\text{EI}, \mu^\text{IR}$ & Average incubation and infectious periods; $\sigma^{-1}$ and $\gamma^{-1}$ in the basic model.\\
$W^\text{EI}, W^\text{IR}$ & Random variables for latent and infectious periods \\
$\pi^\text{EI}, \pi^\text{IR}$ & Delay distributions: $E{\to}I$, $I{\to}R$ \\
\bottomrule
\end{tabularx}
\end{table}

\section{Background}\label{sec:background}

Literature on reproduction numbers dates back decades, especially on the basic reproduction number
\citep{dublin1925true, kermack1927contribution, macdonald1952equilibrium, diekmann1990definition, anderson1991infectious, heesterbeek2002brief}.
This metric, often referred to as $R_0$, is the expected number of secondary infections produced by one infectious individual in a fully susceptible population.
It satisfies a sharp threshold: epidemics grow when $R_0 > 1$ and die out when $R_0 < 1$.
Other important heuristics can be approximated from $R_0$, including short- and long-term forecasts, final size analysis, and the herd immunity threshold.

Since the assumption of full susceptibility rarely holds beyond an outbreak's earliest moments, attention often shifts to the effective reproduction number $R_t$.
Colloquially, $R_t$ is the expected secondary infections from an individual infected at time $t$ given the current population state.
\citet{fraser2007} operationalizes this definition with two related formalizations, which are conducive to estimation.
While his ideas were originally presented in continuous time, we define them discretely for simplicity of explanation.

Let $x_t$ be the number of individuals infected at $t$.
This indexes on the date the infection transmits, as opposed to when the person becomes infectious (capable of transmitting to others).
Further define $w(t,k)$ as the rate of secondary transmissions for a primary infection of age $k$ at time $t$:
$$w(t,k)=\mathbb{E}[\text{\# transmissions at $t$}\given\text{individual infected at $t-k$}].$$
This rate defines the mean process that underlies the transmission model:
\begin{equation}\label{eq:fraser_alt}
    \mathbb{E}[x_t\given x_{< t}] = \sum_{k> 0} x_{t-k} \cdot w(t,k),
\end{equation}
where in general $x_{<t} = (x_0,\dots,x_{t-1})$.

Some works define $w(t,k)$ to satisfy an exact transmission model, dropping the expectation in \eqref{eq:fraser_alt} \citep{fraser2007, abbott2020estimating}.
However, when infection counts $x_t$ are random, this empirical $w(t,k)$ must shift accordingly, implying $R_t$ is nonsmooth over time.
To avoid this, we treat the rate $w(t,k)$ as an expectation, per works like \citet{cori2013new}.

Before presenting $R_t$ definitions, we introduce the generation interval distribution $g$.
The generation interval is the age of a primary infection at the time that it transmits to a second individual.
In discrete time, each element $g_k$ is the probability that transmission occurs $k$ timesteps after infection, given it occurs at all.
\citet{fraser2007} formalizes $g_k$ as the $k$-step transmission rate normalized by its total mass \eqref{eq:inst}:
\begin{align}\label{eq:gi_basic}
    g_k&=\frac{w(t,k)}{\sum_j w(t,j)}\\
    &=\mathbb{P}(\text{transmission occurs $k$ steps after infection}\given\text{transmission occurs}).
\end{align}
\subsection{Case and Instantaneous \texorpdfstring{$R_t$}{Rt}}

The case reproduction number (also called cohort $R_t$) is defined as
\begin{align}
    R_t^C &= \sum_k w(t+k,k) \\
    &\approx \text{Average secondary infections from cohort of primary infections at $t$.}
\end{align}

Case $R_t$ is a forward-looking quantity.
It reports conditions at time $t$ by how many secondary events (transmissions) they generate in the future.

While straightforward, case $R_t$ may be more descriptive of the near future than the present \citep{gostic2020practical}.
To analyze an intervention, $R_t^C$ has the undesirable property of changing before it occurs, since it affects individuals who were already infected.
This complicates interpretation, as it erroneously suggests transmission was already falling before the imposition of the intervention.

While we study case $R_t$ further in \cref{apx:case_rt}, this paper focuses more on the other definition in \citet{fraser2007}.

Case $R_t$ aggregates $k$-step transmission rates from $t$ onwards.
As an alternative, instantaneous $R_t$ shifts these rates backwards in time:
\begin{align}\label{eq:inst}
    R_t^I &= \sum_k w(t,k) \\
    &\approx \text{Average secondary infections at $t$ from a recent primary infection.}
\end{align}
Unlike case $R_t$, this rate is inherently backward-looking, in that it does not consider transmission after $t$.
This makes it more suitable for understanding present conditions, such as the effect of interventions.
Note that case and instantaneous $R_t$ are the same if \mbox{$w(t+k,k)=w(t,k)\;\forall k$}, meaning if $k$-step transmission rates are equivalent before and after $t$.
Thus, we can also interpret instantaneous $R_t$ as being the average number of secondary transmissions per primary infection at time $t$ assuming conditions remain the same into the near future.

Plugging \eqref{eq:inst} into the definition of generation interval \eqref{eq:gi_basic} yields \mbox{$g_k=w(t,k)/R_t^I$}.
Substituting this relation into the transmission model \eqref{eq:fraser_alt} produces the renewal equation:
\begin{equation}\label{eq:renewal}
    \mathbb{E}[x_t\given x_{< t}] = R_t^I \sum_{k> 0} x_{t-k} \cdot g_k,
\end{equation}
The renewal equation may be rearranged to solve for $R_t^I$, with $\mathbb{E}[x_t\given x_{< t}] $ in the numerator.

Many works assume a stochastic noise model for infections $x_t$.
For example, \citet{cori2013new} assumes $x_t$ is drawn from a Poisson distribution whose rate is parameterized by $R_t^I$.
Their Bayesian method, EpiEstim, learns a posterior on $R_t$ with this likelihood and a gamma prior.
They propose using case reports as a surrogate for infections $x_t$, which are latent (unobserved).

\subsection{Mechanistic \texorpdfstring{$R_t$}{Rt}}
\label{sec:compartmental-bg}

We now formalize the compartmental perspective on $R_t$
\citep{anderson1991infectious}.
Compartmental models partition a population of size $N$ into
groups representing disease status, with transitions between
groups governed by specified rates.
In the original SIR model, susceptible individuals move to the infectious ($I$) compartment upon transmission, before recovering ($R$) after some duration \citep{kermack1927contribution}.
The $R$ compartment is also sometimes called Removal to include death.
We focus on the SEIR model, which extends this framework
with an Exposed compartment ($E$) for individuals who are infected
but not yet capable of transmitting the disease to others.
This compartment is biologically meaningful for pathogens like influenza and COVID-19 with non-negligible
latent periods.

The standard SEIR model is a deterministic, continuous-time process governed by a set of differential equations.
These equations contain three parameters, which are conventionally treated as stationary.
Most important for our purposes is $\beta$, the contact rate.
$\beta$ is the average number of contacts per unit time for an infectious individual, multiplied by the probability that a susceptible-infectious contact transmits the disease.
In addition, exposed individuals progress to infectious at rate $\sigma$, and infectious individuals recover at rate $\gamma$.
Implicitly, these assume waiting times are exponentially distributed, with mean latent and infectious periods of $1/\sigma$ and $1/\gamma$.
We will later relax this assumption.

The basic reproduction number $R_0$ supposes the entire population is susceptible.
In a compartmental model, this is true at the onset of the epidemic, $t=0$.
There, $\beta$ is the daily number of transmissions per susceptible.
$R_0$ is the lifetime number of secondary transmissions, so its formula merely scales $\beta$ by the average duration of infectiousness.
In simple compartmental models this is $1/\gamma$, so \mbox{$R_0 = \beta/\gamma$}.

In contrast, the effective reproduction number describes conditions at time $t$---particularly, the fact that not all individuals are susceptible.
Prevalence $S_t$ decays over the course of the epidemic, from an initial value around $N$.
Therefore, mechanistic $R_t$ in its simplest form is
\begin{equation}\label{eq:mech-rt-basic}
    R_t^M = \frac{\beta S_t }{\gamma N}.
\end{equation}

Mechanistic $R_t$ is the average number of secondary transmissions over the lifetime of an \textit{infectious} individual at time $t$, assuming conditions remain the same.
(In reality, conditions will inevitably change as susceptibles $S_t$ deplete, but this is a small effect.)
This notion of active infectiousness differentiates $R_t^M$ from the $R_t^C$ and $R_t^I$, which marked individuals based on their time of infection.
We compare them in greater depth in the next section.

This definition of mechanistic $R_t$ is applicable to essentially any compartmental model.
We next present several ways in which the basic model presented above may be modified for practical purposes.

One important variation allows the rate parameters to be nonstationary.
This is particularly useful for $\beta$, as behavior and policy dynamics affect the contact rate.
We therefore depart from convention and allow $\beta_t$ to vary with time.
We call $\beta_t$ the effective contact rate, akin to the effective reproduction number $R_t$.
Without losing generality, we still assume the parameters that underlie waiting times are constant, since they are primarily biological.

Compartmental models can also be expressed in discrete-time.
Epidemic surveillance data typically arrives with daily resolution, so
the continuous-time model is often approximated with an Euler step.
This first-order approximation uses day-to-day changes in each compartment's prevalence in place of its derivative.
These changes can be attributed to both inflow and outflow;
we denote inflow ``incidence'' with asterisks.
For example, susceptible depletion can be written in terms of incident exposures $E_t^*$:
\begin{equation}
S_{t+1}=S_t - E_t^* \approx S_t + \frac{dS}{dt},\quad\text{where}\quad E_t^* \approx -\frac{dS}{dt} = \beta_t \frac{S_t I_t}{N}.
\end{equation}
Note that $E_t^*$ corresponds to $x_t$ in the prior section.

Mechanistic $R_t$ can also be revised for non-exponential waiting times between compartments.
This canonical assumption produces convenient differential equations: $I_t^* = \sigma E_t$ for infectious incidence and $R_t^* = \gamma I_t$ for recoveries.
However, the delay distributions they imply may have unrealistic variance, often too high.
Instead, define the exposure-to-infectious distribution $\pi^\text{EI}$ such that \mbox{$I_t^* = \sum_{s<t} E_s^* \pi_{t-s}^\text{EI}$} in a deterministic model;
Similarly define $\pi^\text{IR}$, the infectious-to-recovery distribution that relates \mbox{$R_t^* = \sum_{s<t} I_s^* \pi_{t-s}^\text{IR}$}.
This distribution has mean $\mu^\text{IR}$, which was $1/\gamma$ in the exponential case.
Combining $\mu^\text{IR}$ with $\beta_t$ yields our generalized definition of mechanistic $R_t$:
\begin{equation}\label{eq:mech-rt-cp}
    R_t^M = \beta_t\cdot \frac{S_t \mu^\text{IR}}{N}.
\end{equation}

Compartmental models may be stochastic, not deterministic.
Poisson models are the standard choice \citep{andersson2000stochastic},
as modeling exponential transitions as random implies a Poisson process for the number of events per unit time.
We impose this in our framework, though other noise models are possible.
The binomial distribution is a natural alternative, per the Reed-Frost model \citep{abbey1952examination}.
However, Le Cam's theorem shows that this can be well-approximated by a Poisson, when the population is large and per-timestep transition probabilities are small.

Putting these pieces together---time-varying $\beta$, discrete-time approximation, arbitrary delay distributions, and stochastic noise---we arrive at the following mechanistic model:
\begin{subequations}\label{eq:seir-model}
\begin{alignat}{2}
    S_{t+1} &= S_t - E_t^*,         &\qquad E_t^* &\sim \text{Pois}\left(\beta_t \frac{S_t I_t}{N}\right) \label{eq:seir-s}\\
    E_{t+1} &= E_t + E_t^* - I_t^*, &\qquad I_t^* &\sim \text{Pois}\left(\sum_{s<t} E_s^* \pi_{t-s}^\text{EI}\right) \label{eq:seir-e}\\
    I_{t+1} &= I_t + I_t^* - R_t^*, &\qquad R_t^* &\sim \text{Pois}\left(\sum_{s<t} I_s^* \pi_{t-s}^\text{IR}\right) \label{eq:seir-i}\\
    R_{t+1} &= R_t + R_t^*.         &              & \label{eq:seir-r}
\end{alignat}
\end{subequations}
Henceforth, we use \cref{eq:seir-model} when discussing mechanistic $R_t$ and its estimation.

Finally, compartments can be added or dropped without changing the meaning of $R_t^M$.
The crucial piece is the SIR backbone, which encompasses $\beta_t$ and $\mu^\text{IR}$.
Nevertheless, we introduce mechanistic $R_t$ in terms of SEIR models because they are more justifiable than SIR.
A common extension is the SEIRD model, adding a fifth compartment for death ($D$).

\section{Equivalence of Instantaneous and Mechanistic \texorpdfstring{$R_t$}{Rt}}\label{sec:equivalence}

We now study the relationship between the definitions of $R_t$ introduced in \cref{sec:background}.
Our main result is as follows.

\begin{proposition}\label{prop:equivalence}
    Assume a population of size $N$ with homogenous mixing, meaning all contacts are equally likely to come into contact with one another.
    Then instantaneous $R_t$ and mechanistic $R_t$ are equivalent: $R_t^I=R_t^M$.
\end{proposition}

Recall $R_t^I = \sum_{k> 0} w(t,k)$ depends on how the $k$-step transmission rate $w(t,k)$ is defined.
We use the mean interpretation from \eqref{eq:fraser_alt}, as opposed to a data-driven notion that drops the expectation.
Otherwise, the $R_t$ equivalence would only hold in expectation.

\begin{proof}

Expanding on the definition of $w(t,k)$,
\begin{align}
w(t,k) &= \mathbb{E}\left[\text{\# secondary transmissions at $t$}\mid \text{primary infection at $t-k$}\right]\\
    &=\mathbb{E}\left[\frac{1}{E_{t-k}^*}\sum_{i=1}^{E_{t-k}^*} \text{\# transmissions at $t$ of $i$'th infection at $t-k$}\right] \\
    &= \mathbb{E}\left[\frac{1}{E_{t-k}^*}\sum_{i=1}^{E_{t-k}^*} \sum_{j=1}^{S_{t}}\mathbf{1}\{\text{infection $i$ transmits to susceptible } j \text{ at } t\}\right]\\
    &= S_t \cdot \mathbb{P}(\text{$i$'th infection at $t-k$ transmits to $j$'th susceptible at $t$})
\end{align}

Three conditions must be met for transmission to occur.
First, the $i$'th infection at $t-k$  must still be infectious by $t$.
Since delay distributions are assumed stationary, this is equivalent to an infection being infectious $k$ days later.
Secondly, they must come into contact with the $j$'th susceptible.
Finally, secondary transmission must actually occur, as not all susceptible-infectious contacts result in infection.
Formalizing this and summing over all lags $k$  yields
\begin{align}
    R_t^I &= S_t \sum_{k> 0} {\mathbb{P}\!\left(\text{infectious at } t \mid \text{new infection at } t-k\right)} \\
    &\hspace{70pt}\cdot {\mathbb{P}(i \text{ contacts } j \text{ at } t)} \cdot {\mathbb{P}
    (\text{transmission} \mid S\text{-}I \text{ contact})}\\
     &= S_t \cdot {\mathbb{P}(i \text{ contacts } j \text{ at } t)} \cdot {\mathbb{P}
    (\text{transmission} \mid S\text{-}I \text{ contact})}\\
    &\hspace{70pt}\cdot\sum_{k> 0} {\mathbb{P}\!\left(\text{infectious $k$ days after infection}\right)}.
\end{align}

Next, we leverage the following two lemmas, proven below.
\begin{lemma}\label{lem:infectious_kernel}
For any delay distributions $\pi^\text{EI}$ and $\pi^\text{IR}$,
    $$\sum_{k> 0} {\mathbb{P}\!\left(\text{infectious $k$ days after infection}\right)} = \mu^\text{IR}.$$
\end{lemma}

\begin{lemma}\label{lem:homo_to_beta}
    Assuming homogeneous mixing,
    $${\mathbb{P}(i \text{ contacts } j \text{ at } t)} \cdot {\mathbb{P}
    (\text{transmission} \mid S\text{-}I \text{ contact})} = \beta_t / N.$$
\end{lemma}

Applying these lemmas, $R_t^I = S_t \cdot  \mu^\text{IR} \cdot \beta_t/N$.
By definition, this is $R_t^M$  \eqref{eq:mech-rt-cp}, completing the proof.
Again, $\mu^\text{IR} = \gamma^{-1}$ in the special case that $W^\text{IR} \sim \mathrm{Exp}(\gamma)$ or $\mathrm{Geom}(\gamma)$.
\cref{sec:compartmental-generation} provides further connections between instantaneous and mechanistic $R_t$.

\end{proof}

\subsection{Proof of \texorpdfstring{\cref{lem:infectious_kernel}}{Lemma \ref*{lem:infectious_kernel}}}

We want to show that the probability of being infectious $k$ days after infection sums to $\mu^\text{IR}$, the mean duration of infectiousness.
We define this probability as
\begin{equation}\label{eq:infectious-kernel}
\zeta^\text{EI}_k := \Pprob(\text{infectious $k$ days after infection}).
\end{equation}
In an SIR model, this reduces to the survival function of $W^\text{IR}$.
For SEIR models, a person is infectious $k$ days after infection if their latent period $W^\text{EI}$ was some $j\in\{1,\dots,k\}$ and their infectious duration $W^\text{IR}$ has not yet elapsed at age $k$.
By the law of total probability and conditional independence of $W^\text{EI}$ and $W^\text{IR}$,
    \begin{align}
    \zeta^\text{EI}_k
    &= \sum_{j=1}^{k} \Pprob(W^\text{EI}=j) \cdot \Pprob(W^\text{IR} > k-j)\\
    &= \sum_{j=1}^{k} \pi_j^\text{EI} \cdot \sum_{\ell=k-j+1} \pi_\ell^\text{IR}
\end{align}

To show that $\sum_k \zeta^\text{EI}_k=\mu^\text{IR}$, swap the order of summation:

\begin{align*}
    \sum_{k \geq 1} \zeta^\text{EI}_k &= \sum_{k \geq 1} \sum_{j=1}^{k} \Pprob(W^\text{EI} = j) \cdot \Pprob(W^\text{IR} > k-j)\\
    &= \sum_{j \geq 1} \Pprob(W^\text{EI} = j) \sum_{k \geq j} \Pprob(W^\text{IR} > k-j).
\end{align*}
The inner sum runs over $k = j, j+1, \dots$; substitute $a = k - j $ so that $a = 0,1, 2, \dots$:
\[
\sum_{k \geq j} \Pprob(W^\text{IR} > k-j) = \sum_{a \geq 0} \Pprob(W^\text{IR} > a).
\]
This is independent of $j$, so it factors out of the outer sum:
\[
\sum_{k \geq 1} \zeta^\text{EI}_k = \sum_{j \geq 1} \Pprob(W^\text{EI} = j) \cdot \sum_{a \geq 0} \Pprob(W^\text{IR} > a) = 1\cdot \mu^\text{IR}.
\]
The last line invokes the total mass of a distribution and the fact that the survival function of a non-negative random variable sums to its mean.

\subsection{Proof of \texorpdfstring{\cref{lem:homo_to_beta}}{Lemma \ref*{lem:homo_to_beta}}}

Define $C_t$ as a random variable for the average number of contacts at time $t$.
Define its mean:
\[
\mu_t = \mathbb{E}[C_t] = \sum_{n=0}^N n \cdot P(C_t = n).
\]
Under the assumption of homogeneous mixing, all contacts are equally likely, and the fact that $i$ is infectious and $j$ is irrelevant.
We decompose the contact probability using the law of total probability:
\begin{align}
     \mathbb{P}(\text{Infectious $i$ contacts Susceptible $j$ at $t$}) &=  \mathbb{P}(\text{$i$ contacts $j$ at $t$})\\
     &= \sum_{n=0}^N \mathbb{P}(\text{$i$ contacts $j$ at $t$}\given C_t=n)\Pprob(C_t=n)\\
     &= \sum_{n=0}^N \frac{n}{N}\, P(C_t = n) = \frac{\mu_t}{N}.
\end{align}

Finally, recall the effective contact rate is defined as  \mbox{$\beta_t = \mu_t \cdot \Pprob(\text{transmission} \mid S\text{-}I \text{ contact at } t)$}. Hence,
\begin{align}
     &\mathbb{P}(\text{Infectious $i$ contacts Susceptible $j$ at $t$})\cdot
\mathbb{P}(\text{transmission} \mid S\text{-}I \text{ contact}) \\&\hspace{50pt}= \mu_t/N \cdot \mathbb{P}(\text{transmission} \mid S\text{-}I \text{ contact at }t) \\
&\hspace{50pt}= \frac{\beta_t}{N}.
\end{align}

\section{The SEIR Generation Interval}\label{sec:compartmental-generation}

The equivalence result established that instantaneous and mechanistic $R_t$ agree under homogeneous mixing.
We now ask a complementary question: what generation interval distribution does the SEIR model imply?
Deriving this closes the loop between the two frameworks, expressing the renewal equation entirely in compartmental terms.
\begin{proposition}\label{prop:compartmental-generation}
    Recall the infectious kernel $\zeta^\text{EI}_k = \Pprob(\text{infectious $k$ days after exposure})$ \eqref{eq:infectious-kernel} for a discrete-time compartmental model. The generation interval distribution, which satisfies the renewal equation, has weights
    $$g_k := \Pprob(\text{transmission $k$ days after exposure $\vert$ transmission occurs}) = \zeta^\text{EI}_k/{\mu^{IR}}.$$
    Moreover, $\mu^\text{IR}$ is the total mass of the kernel $\zeta^\text{EI}$, so $g_k = \frac{\zeta_k^\text{EI}}{\sum_j \zeta_j^\text{EI}}.$
    Lastly, the mean generation time (proved in \cref{apx:mean-generation}) is
    \[
        \mathbb{E}[G] = \mu^\text{EI} + \frac{\mu^\text{IR}}{2} - \frac{1}{2} + \frac{(\sigma^\text{IR})^2}{2\, \mu^\text{IR}}.
    \]
\end{proposition}

\begin{proof}
    For now, suppose transmissions are deterministic.
In an SEIR model, $E_t^* = \beta_t \cdot \frac{S_t}{N} \cdot I_t$.
Instantaneous $R_t$ is equivalent to mechanistic $R_t$, defined as $R_t^M = \beta_t \cdot \frac{S_t}{N} \cdot \mu^{IR}$;
plugging this in yields $E_t^* = R_t^I \cdot I_t / \mu^\text{IR}$.
Infectious prevalence $I_t$ can be defined using the kernel $\zeta^\text{EI}$:
\begin{equation}
I_t = \sum_{k \geq 1} E_{t-k}^* \cdot \zeta^\text{EI}_k.
\label{eq:prev-decomp}
\end{equation}
Substituting this expression in for $I_t$ yields the renewal equation \eqref{eq:renewal}, expressed in compartmental terms:
\begin{equation}\label{eq:mech-renewal-halfway}
    E_t^* = R_t^I/ \mu^\text{IR} \cdot \sum_{k \geq 1} E_{t-k}^* \cdot \zeta^\text{EI}_k =  R_t^I \cdot \sum_{k \geq 1} E_{t-k}^* \cdot (\zeta^\text{EI}_k / \mu^\text{IR}).
\end{equation}
Comparing with \eqref{eq:renewal}, the generation interval weights are $g_k=\zeta^\text{EI}_k / \mu^\text{IR}$:
the infectious kernel $\zeta^\text{EI}$, rescaled by the mean infectious duration $\mu^\text{IR}$.

To handle the stochastic case,
recall that instantaneous $R_t$ and the generation interval are typically defined taking the renewal equation in expectation: \mbox{$\mathbb{E}[E_t^*\ \vert\ E_{s<t}^*] = \sum_{k\geq 1} E_{t-k}^* g_k$}
To derive $\mathbb{E}[E_t^*\ \vert\ E_{s<t}^*]$ in the compartmental perspective, we use the tower rule to impute $I_t$:
\begin{equation}
    \mathbb{E}[E_t^*\ \vert\ E_{s<t}^*] = \mathbb{E}[\mathbb{E}[E_t^*\ \vert\ E_{s<t}^*, I_t]]  =
\frac{R_t}{\mu^\text{IR}}
\mathbb{E}[ I_t\ \vert\ E_{s<t}^*]
= \frac{R_t}{\mu^\text{IR}}
\sum_{k \geq 1} E_{t-k}^* \cdot \zeta^\text{EI}_k.
\end{equation}
Thus, we return to the renewal equation, with $g_k = \zeta_k^\text{EI}/\mu^\text{IR}$ defined probabilistically.
\end{proof}

\subsection{Discussion}

Renewal equation estimators are often presented as relatively assumption-light: given a generation interval $g$ and observed infections, one can estimate $R_t$ without specifying a full mechanistic model.
However, the results above reveal that the generation interval itself encodes compartmental structure.
The weights $g_k = \zeta_k^\text{EI}/\mu^\text{IR}$ depend on the latent and infectious period distributions $\pi^\text{EI}$ and $\pi^\text{IR}$, and the mean generation time depends on their first two moments.
In practice, choosing a generation interval for use in methods like EpiEstim implicitly assumes a model of disease progression, even if that model is never written down.
Making this dependence explicit, as we have done here, clarifies what assumptions are baked into standard $R_t$ estimates.

Our equivalence result relies on homogeneous mixing: all individuals are equally likely to contact one another.
Real populations exhibit substantial heterogeneity in contact patterns across age groups, geographic regions, and behavioral strata.
Extending the renewal--compartmental correspondence to structured populations is a natural direction, though in practice most $R_t$ estimation pipelines assume homogeneous mixing as a simplifying approximation \citep{gostic2020practical}.

\bibliographystyle{plainnat}
\bibliography{refs}

\appendix

\section{Proof of mean generation time for SEIR}\label{apx:mean-generation}

We prove the mean generation time stated in \cref{prop:compartmental-generation}.
\begin{proposition}[Restated]
The mean generation time for an SEIR model is
    \[
        \mathbb{E}[G] = \mu^\text{EI} + \frac{\mu^\text{IR}}{2} - \frac{1}{2} + \frac{(\sigma^\text{IR})^2}{2\, \mu^\text{IR}}.
    \]
\end{proposition}

\begin{proof}

We now derive the mean generation time $G$. By definition,
\[
\mathbb{E}[G] = \sum_{k \geq 1} k\, g_k = \frac{1}{\mu^\text{IR}} \sum_{k \geq 1} k\, \zeta^\text{EI}_k.
\]
Substituting $\zeta^\text{EI}_k = \sum_{j=1}^{k} \pi^\text{EI}_j \cdot \Pprob(W^\text{IR} > k-j)$ and swapping the order of summation,
\[
\sum_{k \geq 1} k\, \zeta^\text{EI}_k
= \sum_{k \geq 1} \sum_{j=1}^{k} k\, \pi^\text{EI}_j\, \Pprob(W^\text{IR} > k-j)
= \sum_{j \geq 1} \pi^\text{EI}_j \sum_{k \geq j} k\, \Pprob(W^\text{IR} > k-j).
\]
Reindex the inner sum with $a = k - j$, so that $a = 0, 1, 2, \dots$ and $k = j + a$:
\[
\sum_{j \geq 1} \pi^\text{EI}_j \sum_{a \geq 0} (j + a)\, \Pprob(W^\text{IR} > a)
= \sum_{j \geq 1} \pi^\text{EI}_j \left[ j \sum_{a \geq 0} \Pprob(W^\text{IR} > a) + \sum_{a \geq 0} a\, \Pprob(W^\text{IR} > a) \right].
\]
The first inner sum is $\mu^\text{IR}$, since the survival function of a non-negative integer random variable sums to its mean. The bracketed expression is therefore $j\, \mu^\text{IR} + S$, where $S := \sum_{a \geq 0} a\, \Pprob(W^\text{IR} > a)$ does not depend on $j$. Pulling $S$ out of the outer sum and using $\sum_j j\, \pi^\text{EI}_j = \mu^\text{EI}$,
\[
\sum_{k \geq 1} k\, \zeta^\text{EI}_k = \mu^\text{EI}\, \mu^\text{IR} + S.
\]

It remains to evaluate $S$. Writing the survival function as a sum over $\pi^\text{IR}$ and again swapping the order of summation,
\[
S = \sum_{a \geq 0} a \sum_{r > a} \pi^\text{IR}_r
= \sum_{r \geq 1} \pi^\text{IR}_r \sum_{a = 0}^{r-1} a.
\]
The inner sum is the standard arithmetic series $\sum_{a=0}^{r-1} a = \tfrac{r(r-1)}{2} = \tfrac{r^2 - r}{2}$. Therefore
\begin{align*}
    S &= \sum_{r \geq 1} \pi^\text{IR}_r \cdot \frac{r^2 - r}{2}
= \frac{1}{2}\left( \sum_r r^2\, \pi^\text{IR}_r - \sum_r r\, \pi^\text{IR}_r \right)
= \frac{1}{2}\left( \mathbb{E}[(W^\text{IR})^2] - \mu^\text{IR} \right)\\
&=\frac{1}{2}\left((\mu^\text{IR})^2 + (\sigma^\text{IR})^2 - \mu^\text{IR}\right).
\end{align*}
Combining the two results and dividing by $\mu^\text{IR}$,
\[
\mathbb{E}[G] = \frac{\mu^\text{EI}\, \mu^\text{IR} + S}{\mu^\text{IR}} = \mu^\text{EI} + \frac{\mu^\text{IR}}{2}  + \frac{(\sigma^\text{IR})^2}{2\, \mu^\text{IR}} - \frac{1}{2}.
\]
\end{proof}

Notably, the latent-period SD $\sigma^\text{EI}$ does not enter the expression for $\mathbb{E}[G]$. The $-\tfrac{1}{2}$ reflects that $\zeta^\text{EI}$ assigns mass at the day of the E$\to$I transition itself ($\Pprob(W^\text{IR} > 0) = 1$), so an individual is permitted to transmit on the very day they enter the infectious compartment. It vanishes in the continuous-time analogue.

\section{Case reproduction numbers}\label{apx:case_rt}

Similar arguments deepen our understanding of how case $R_t$ relates to the other definitions.
\citet{fraser2007} relates $R_t^I$ and $R_t^c$ through the generation interval distribution by substituting $w(t,k)=R_t^Ig_k$ into the definition of case $R_t$:
\begin{equation}\label{eq:case-inst-fraser}
    R_t^C = \sum_{k> 0} R_{t+k}^I g_k.
\end{equation}
This further reinforces the idea that instantaneous $R_t$ is the average number of secondary transmissions from a primary infection at $t$ if conditions remain unchanged.
We can derive this same expression under compartmental assumptions, using the tools introduced in \cref{prop:compartmental-generation}.

By definition, $R_t^C = \sum_{k > 0} w(t+k,k)$. Expanding $w(t+k,k)$ as in the proof of Proposition~\ref{prop:equivalence} and recalling the definition of $\zeta_k^\text{EI}$,
\begin{align*}
R_t^C
  &=  \sum_{k> 0} S_{t+k}\cdot {\mathbb{P}\!\left(\text{infectious at } t+k \mid \text{new infection at } t\right)} \\
    &\hspace{70pt}\cdot {\mathbb{P}(i \text{ contacts } j \text{ at } t+k)}\cdot\mathbb{P}
    (\text{transmission} \mid S\text{-}I \text{ contact})\\
&= \sum_{k > 0} S_{t+k} \cdot \frac{\beta_{t+k}}{N} \cdot \zeta_k^\text{EI} ,
\end{align*}
\cref{lem:homo_to_beta} established that instantaneous $R_t$ and mechanistic $R_t$ both equal $\beta_t S_t \mu^\text{IR}/N$.
Plugging that in,
$$R_t^C = \sum_{k > 0} R_{t+k}^I \cdot \frac{\zeta_k^\text{EI}}{\mu^{\mathrm{IR}}}.$$
\cref{prop:compartmental-generation} established that in compartmental equations, \mbox{$g_k = \frac{\zeta_k^\text{EI}}{\mu^{\mathrm{IR}}}$}.
This completes the proof of \eqref{eq:case-inst-fraser}.

The method by \citet{wallinga_teunis} is the standard approach to estimate case $R_t$ from aggregate data.
Given individual-level data, one could compute $R_t^C$ by simply identifying the cohort of people infected at $t$ and tracking their secondary transmissions.

\end{document}